\documentclass[preprint,12pt]{elsarticle}

\usepackage{amsmath,amsfonts,amsthm,amssymb,paralist,subfigure,graphicx,amsbsy,float,epsfig,color}
\usepackage{cuted,mathtools,lipsum}

\usepackage[ruled,vlined,linesnumbered]{algorithm2e}
\usepackage{stmaryrd,url}

\usepackage{amsthm}
\newtheorem{lem}{Lemma}
\newtheorem{ass}{Assumption}
\newtheorem{theorem}{Theorem}

\newtheorem{rem}{Remark}

\def\mb{\mathbf}

\def\mc{\mathcal}

\def\mb{\mathbf}

\def\mc{\mathcal}

\DeclareMathOperator*{\argmin}{argmin}

\journal{Journal of The Franklin Institute}

\begin{document}

\begin{frontmatter}

\title{Momentum-based Accelerated Algorithm for Distributed Optimization under Sector-Bound Nonlinearity
}

\author[Sem]{Mohammadreza Doostmohammadian}
\affiliation[Sem]{Mechatronics Group, Faculty of Mechanical Engineering, Semnan University, Semnan, Iran, doost@semnan.ac.ir.}
\author[SP]{Hamid R. Rabiee}
\affiliation[SP]{Computer Engineering Department, Sharif  University of Technology, Tehran, Iran, rabiee@sharif.edu}

\begin{abstract}
Distributed optimization advances centralized machine learning methods by enabling parallel and decentralized learning processes over a network of computing nodes.  This work provides an accelerated consensus-based distributed algorithm for locally non-convex optimization using the gradient-tracking technique. The proposed algorithm (i) improves the convergence rate by adding momentum towards the optimal state using the heavy-ball method, while (ii) addressing general sector-bound nonlinearities over the information-sharing network. The link nonlinearity includes any sign-preserving odd sector-bound mapping, for example, log-scale data quantization or clipping in practical applications. For admissible momentum and gradient-tracking parameters, using perturbation theory and eigen-spectrum analysis, we prove convergence even in the presence of sector-bound nonlinearity and for locally non-convex cost functions. Further, in contrast to most existing weight-stochastic algorithms, we adopt weight-balanced (WB) network design. This WB design and perturbation-based analysis allow to handle \textit{dynamic} directed network of agents to address possible time-varying setups due to link failures or packet drops.
\end{abstract}

%%Graphical abstract
\begin{graphicalabstract}
	\includegraphics{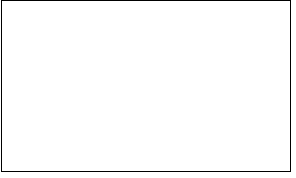}
\end{graphicalabstract}

%%Research highlights
\begin{highlights}
	\item Introducing a new accelerated algorithm for distributed optimization based on the gradient-tracking technique.
	\item Addressing general sector-bound nonlinearities in the information-sharing network, including sign-preserving odd mappings such as log-scale data quantization and clipping.
	\item Proposing a weight-balanced network structure, contrasting with existing weight-stochastic algorithms,  to accommodate time-varying conditions such as link failures and packet drops.
\end{highlights}


\end{frontmatter}

\section{Introduction} \label{sec_intro}
Decentralized methods for data mining are of interest in general machine learning \cite{xin2020decentralized}, signal processing \cite{di2020distributed}, and control \cite{camponogara2010distributed} applications. The idea is to parallelize and distribute the learning process and computation over a network of agents. This is motivated by recent advances in intelligent transportation systems (ITS) \cite{shen2022fully,shen2022nonconvex},  Internet-of-Things (IoT) \cite{wu2020collaborate,yu2021jointly}, cloud computing \cite{DOOSTMOHAMMADIAN2025100983}, and distributed cyber-physical-systems (CPS) \cite{isj_cyber}. In these setups, data samples are distributed across multiple agents/nodes, and computation tasks are divided among them. The data sharing between agents only occurs through established communication links where, in reality, might be subject to some non-ideal conditions, for example, due to quantization \cite{oh2021automated,kalantari2010logarithmic,jiang2023jumping} and saturation/clipping \cite{zhang2019gradient}. These nonlinearities make the existing linear methods impractical and inefficient. Such nonlinearities are required to be addressed in the optimization algorithm, otherwise, it may cause large optimality gap and even divergence in practice. Another concern is the convergence rate in real-world applications to reduce the residual as fast as possible. Moreover, in \textit{mobile} multi-agent networks, the communications and linkings might change with data-sharing channels coming and going between different computing nodes. The change in the network might also be due to packet drops \cite{icrom22}. This necessitates algorithms adaptable to time-varying network structures.
Addressing non-ideal data-sharing and dynamic network topology while improving the rate of convergence is both challenging and of great interest in real-world large-scale distributed setups.

The existing literature on distributed optimization and machine learning rarely addresses the nonlinearity of the links over dynamic networks. Most existing literature on distributed optimisation (both constrained and unconstrained) considers ideal links \cite{qureshi2023distributed,yu2023secure,ddsvm,yi2022primal,XU20239096,DEVILLEROS2024106988}. These works give no possibility to address non-ideal linking, such as saturation and quantization. Similarly, the momentum-based strategies in \cite{nguyen2023accelerated,wang2021distributed,lu2020achieving,jin2022momentum,qin2023parallel,xin2019distributed} are designed over linear ideal channels, and the ones in \cite{wang2021distributed,xin2019distributed} only converge over static time-invariant networks. There exist some accelerated algorithms that adopt sign-based nonlinearity to improve the convergence rate, e.g., see \cite{spl24} and \cite{dai2020distributed2,ning2022fixed,liu2022distributed2} for fixed/finite-time solutions. In contrast to these works, quantization and compression of data in networking setups are a must in real-world applications, which is not well-addressed in the mentioned literature. Quantization reduces the precision of data, thereby reducing the number of bits needed to represent it. By mapping a large set of input values to a smaller set of output values, quantization simplifies the data and makes it more compressible. Few works in the literature address distributed optimization under \textit{uniform} quantization that results in certain optimality gap, for example, see \cite{pu2016quantization,bo2024quantization,chen2020distributed} for distributed optimization and \cite{chen2024communication,hanna2021quantization} for distributed machine learning. What is missing in the existing literature is a distributed algorithm to address log-scale quantization, clipping, and general sector-bound nonlinear models over data-sharing channels.

The information sent over the real-world networks might be subject to different nonlinear mappings, such as clipping and quantization. Logarithmic quantization allocates more precision to smaller values and less to larger ones. This is beneficial because, from a system perspective, updates can have a wide dynamic range, but smaller updates are often more critical for fine-tuning the solution. By preserving the precision of smaller values better than larger ones, logarithmic quantization ensures that significant information is not lost \cite{chen2020data}. To be more specific, for the training process in distributed optimization, the gradients can vary in magnitude. By logarithmic quantization one can represent a large range of values with relatively fewer bits; this helps to efficiently encode gradients of different magnitudes without severe loss of information, which is crucial for maintaining the accuracy of gradient-based methods. Example applications are in distributed resource allocation \cite{ojsys}, speech and audio encoding \cite{jayant1984digital}, and image processing \cite{chen2020data}. On the other hand, clipping and saturation are other concerns in real-world applications, for example, ramp-rate-limits for distributed optimization over energy networks \cite{lcss24} and signal clipping in cognitive radio networks \cite{clipping}.

In this paper, (i) a momentum-based distributed optimization algorithm is proposed. The momentum improves (and accelerates) the convergence rate toward the optimal point.  The algorithm includes two parameters, one to tune the gradient tracking rate and the other to tune the momentum rate for faster convergence. We derive the trade-off relation for the admissible range of these two parameters and prove convergence under this relation using perturbation theory, eigen-spectrum analysis, and Lyapunov theory. (ii) Our proof analysis is irrespective of the convexity of the local cost functions. Therefore, the local costs might be non-convex similar to \cite{yi2022primal,wang2021distributed}. (iii) The proposed algorithm converges over time-varying WB directed networks. This finds application over volatile and unreliable networking setups, for example, in the presence of link failures, packet drops, congestion, or general non-ideal networking conditions. Most importantly, (iv) the proposed algorithm addresses general nonlinear mappings on the information-sharing channels representing, for example, log-scale quantization or clipping. This is not addressed in the existing literature, however, it exists in practical real-world distributed setups. Therefore, the existing methods may result in suboptimal solutions or a large optimality gap. Few mentioned works consider \textit{uniform} quantization instead, which, in contrast to log-scale quantization, results in large optimality gap. Our distributed optimization setup shows convergence with the same level of optimality gap both in the presence and absence of sector-bound nonlinearity and advances the state-of-the-art in this regard. Applications of locally non-convex optimization, distributed support-vector-machine (SVM), and distributed linear and logistic regression are simulated to verify the results. The main contributions are summarized as follows:
\begin{enumerate}[(i)]
	\item We propose a distributed algorithm for learning and optimization over large-scale time-varying networks. The algorithm is in single time-scale and computationally more efficient as compared to the double time-scale scenarios.
	\item We introduce a momentum term (known as the heavy-ball method) to accelerate the convergence as compared to the existing literature. 
	\item We further consider possible sector-bound nonlinear mappings on the data exchange links, for example, logarithmic quantization or clipping. This accounts for non-ideal linking conditions in real-world distributed learning setups.  This is not addressed in the existing literature.
	\item We make no assumption on the strict convexity of the objective function, and the proposed algorithm is proved to converge over certain locally non-convex optimization setups, as also shown in the simulations.
\end{enumerate}

The rest of the paper is organized as follows. Section~\ref{sec_prob} states the problem and assumptions. Section~\ref{sec_alg} provides the main momentum-based algorithm and its convergence analysis. Section~\ref{sec_sim} provides the numerical simulations, and Section~\ref{sec_con} concludes the paper.

\textit{Notations:} Table~\ref{tab_notation} summarizes the notations in this paper:
\begin{table}
	\caption{Description of notations and symbols}
	\setlength{\tabcolsep}{0.7\tabcolsep}
	\centering \label{tab_notation}
	\begin{tabular}{ *{2}{c} }
			\hline
			\hline
		\textbf{Symbol} & \textbf{Description} \\
			\hline
		$\mc{G}_\gamma$ & multi-agent network   \\ 
		$\gamma$ & network switching signal   \\ 
		$W_\gamma$ & adjacency matrix of the network \\
		$\overline{W}_\gamma$ & Laplacian matrix\\
		$\mb{x}$ & global state variable  \\
		$\mb{x}_i$ & local state variable at node $i$  \\
		$F$ & global objective function      \\
		$f_i$ & local objective function at node $i$\\
		$f_{i,j}$ & cost/objective for data point $j$ at node $i$ \\
		$n$ & number of nodes/agents\\
		$m_i$ & number of data points at node $i$  \\
		$\nabla f_i$ &  gradient of function $f_i$    \\
		$\nabla^2 f_i$ &  second gradient of function $f_i$    \\
		$H$ &  Hessian matrix  \\
		$\partial_t $ & derivative with respect to time\\
		$\mb{z}_i$ & gradient-tracking auxiliary variable at node $i$   \\
		$\alpha$ & gradient-tracking rate   \\
		$\beta$ & momentum rate \\
		$\overline{\alpha}$ & bound on $\alpha$   \\
		$g_l,g_n$ & nonlinear mapping at the links/nodes  \\
		$I_m$ & identity matrix of size $m$   \\
		$\underline{\kappa}$ & lower sector-bound \\
		$\overline{\kappa}$ & upper sector-bound  \\	
		$t$ & time index  \\						
			\hline \hline
	\end{tabular}
\end{table}

\section{Terminology and Problem Statement} \label{sec_prob}

The problem is to minimize the global objective defined as follows,
\begin{align} \label{eq_prob}
	\min_{\mb{x} \in \mathbb{R}^{m}} &
	F(\mb x) = \frac{1}{n}\sum_{i=1}^{n} f_i(\mb{x}),
\end{align}
with $\mb x  \in \mathbb{R}^{m}$ as the global state variable. This global cost function $F(\cdot) \in \mathbb{R}$ is the sum of local objectives at some computing nodes/agents in the following form,
\begin{align}\label{eq_fij}
	f_i(\mb{x}_i) = \frac{1}{m_i}\sum_{j=1}^{m_i} f_{i,j}(\mb{x}_i),
\end{align}
with $f_{i,j}(\mb{x}_i)$ denoting the cost associated with data point $j$ at node $i$, $m_i$ as the number of data points, and $\mb{x}_i  \in \mathbb{R}^{m}$ as the state variable at node $i$. Obviously, it is assumed that the optimization problem \eqref{eq_prob}-\eqref{eq_fij} has a solution, which is denoted by $\mb{x}^*$ in the rest of the paper.
In this paper, the following assumption holds for the cost functions $f_i(\mb{x}_i)$:
\begin{ass} \label{ass_nonconv}
	Define $H:=\mbox{diag}[\nabla^2 f_i(\mb{x}_i)]$ as the Hessian matrix. The cost functions satisfy the following
	\begin{align} \label{eq_ass1}
		&H \preceq \zeta I_{mn}, \\ \label{eq_ass2}
		& (\mb{1}_n \otimes I_m)^\top H (\mb{1}_n \otimes I_m) \succ 0.
	\end{align}
	The implication is that the local costs $f_i(\mb{x}_i)$ are not necessarily convex.\footnote{An example non-convex model for local cost functions satisfying Assumption~\ref{ass_nonconv} is given in Section~\ref{sec_sim}.}
\end{ass}
The computing nodes (or agents) share data over a communication network $\mc{G}_\gamma$ with adjacency matrix $W_\gamma  =[w^\gamma_{ij}]$. This network $\mc{G}_\gamma$ might be time-varying with its topology changing based on a switching signal $\gamma: t \mapsto \Gamma$, where the set $\Gamma$ includes all the possible topologies for $\mc{G}_\gamma$. The entries of its
Laplacian matrix $\overline{W}_\gamma$ are defined as
\begin{align}
	\overline{W}_\gamma(i,j) = \left\{
	\begin{array}{ll}
		-\sum_{i=1}^{n} w^\gamma_{ij}, & i=j \\
		w^\gamma_{ij}, & i\neq j.
	\end{array}\right.
\end{align}
The following assumption holds for the network $\mc{G}_\gamma$.
\begin{ass} \label{ass_net}
	The directed/undirected network $\mc{G}_\gamma$ is WB, i.e., we have ${W}_\gamma \mb{1}_n={W}_\gamma^\top \mb{1}_n$. It is also strongly-connected at all times.
\end{ass}
\begin{lem} \label{lem_laplac}
	\cite{SensNets:Olfati04,olfatisaberfaxmurray07}
	For a strongly-connected WB network, its Laplacian $\overline{W}_\gamma$ has one isolated zero eigenvalue, with the rest of the eigenvalues in the left-half-plane. Further, the zero eigenvalue is associated with left (and right) eigenvector $\mb{1}_n^\top$ (and~$\mb{1}_n$), i.e., $\mb{1}_n^\top \overline{W}_\gamma= \mb{0}_n$ and~$\overline{W}_\gamma \mb{1}_n=\mb{0}_n$.
\end{lem}
It should be noted that if $\mc{G}_\gamma$ is undirected, the eigenvalues of $\overline{W}_\gamma$ are all real-valued \cite{SensNets:Olfati04}.

\begin{figure}
	\centering
	\includegraphics[width=3.5in]{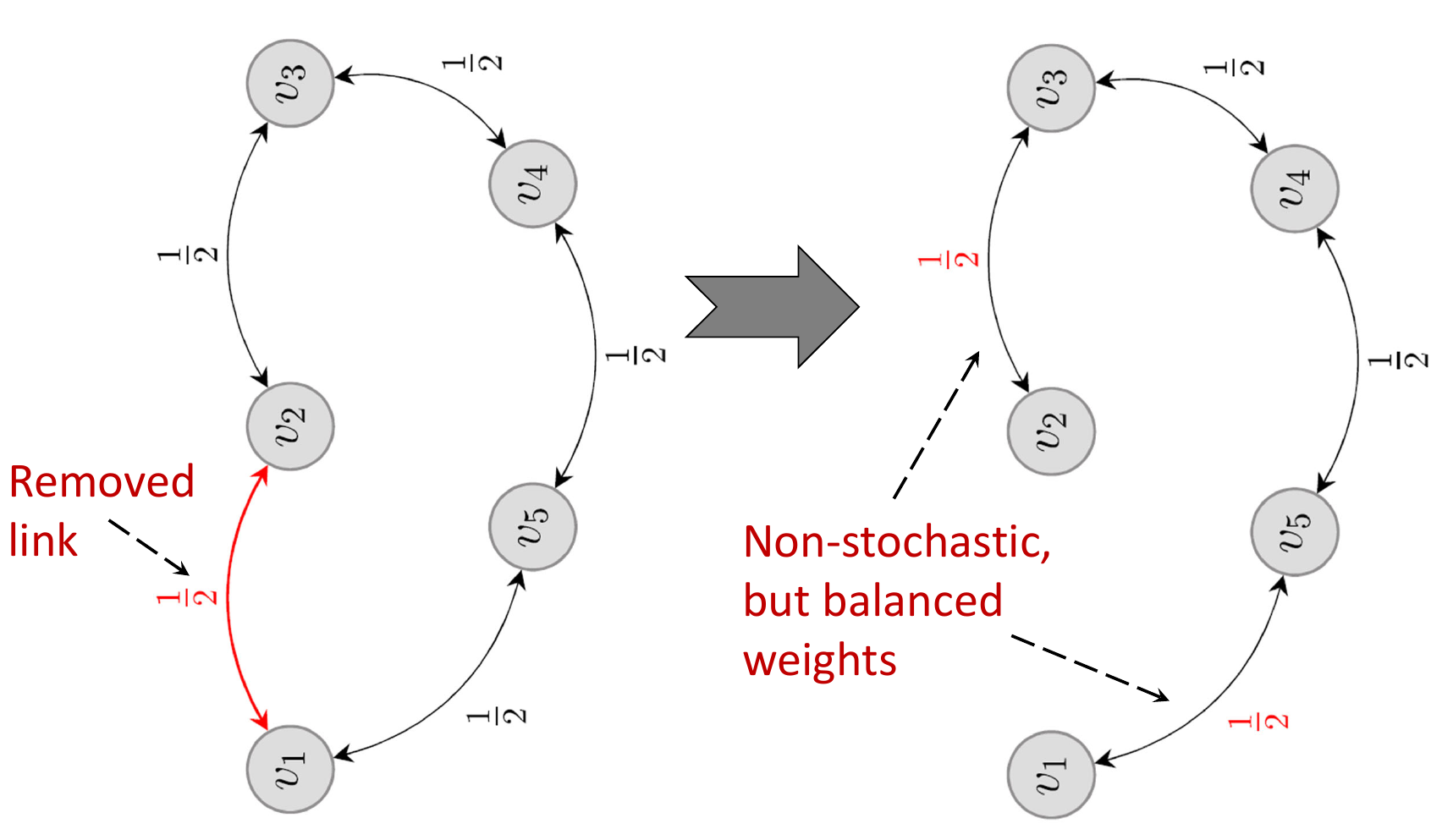}
	\caption{An example graph under link failure: The graph on the left is both weight-balanced and weight-stochastic. After link failure, the graph on the right is still weight-balanced, but is not weight-stochastic. This example shows that weight-balanced design is more robust to link failure as compared to weight-stochastic design.  }  \label{fig_remov}
\end{figure}
\begin{rem}
	In contrast to the stochastic weight design for many existing distributed optimization algorithms, the weight-balanced design is more relaxed and easier to satisfy. In particular, for undirected cases, link failure or change in the topology of weight-stochastic networks necessarily requires redesign algorithms to reset the link weights and reassign stochasticity, see the works by \cite{6426252,cons_drop_siam} for example. However, the weight-balanced condition is already satisfied for undirected networks, and there is no need for extra computations for assigning new weights to the network. This significantly reduces the computational complexity in terms of weight design as compared to the existing weight-stochastic algorithms. An illustrative example is given in Fig.~\ref{fig_remov}.
\end{rem}

The other contribution in this work is to address sector-bound nonlinear models for information-sharing over the network $\mc{G}_\gamma$. This follows the real-world networking applications and distributed systems, in which the data is subject to (logarithmic) quantization or clipping \cite{oh2021automated,kalantari2010logarithmic,jiang2023jumping,zhang2019gradient}. In other words, the information received by the destination agent has typically gone through quantization and/or clipping due to limited communication resources. These techniques are often used to manage the representation and transmission of data efficiently and are used for optimizing bandwidth, storage, and processing capabilities while maintaining acceptable levels of accuracy and performance. The following assumption holds for these nonlinear mappings:
\begin{ass} \label{ass_nonlin}
	The nonlinear mapping $g(\cdot)$ on the data-sharing channels is odd, sign-preserving, and monotonically non-decreasing. Such a nonlinear mapping satisfies the following sector-bound condition:
	\begin{align} \label{eq_sector}
		\underline{\kappa} \mb{x} \leq g(\mb{x}) \leq \overline{\kappa} \mb{x}
	\end{align}
	with $\underline{\kappa}$ and $\overline{\kappa}$ as lower and upper sector-bound rate.
\end{ass}
An example sector-bound nonlinear mapping satisfying Assumption~\ref{ass_nonlin} is shown in Fig.~\ref{fig_sector}. As stated before, logarithmic quantization and (finite-domain) saturation are two nonlinear functions satisfying this condition.
\begin{figure}
	\centering
	\includegraphics[width=2.5in]{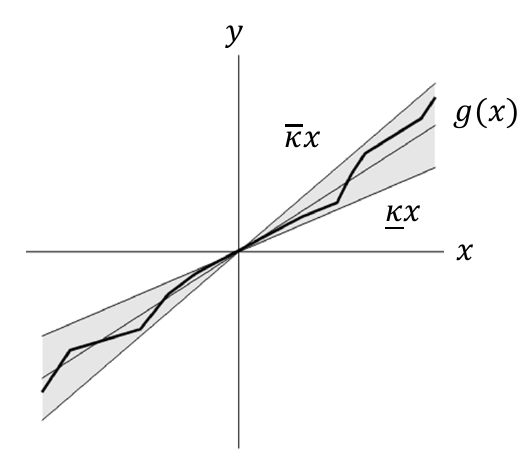}
	\caption{An example nonlinear function $g(\mb{x})$ satisfying the sector-bound condition in Assumption~\ref{ass_nonlin}. The sector-bound region is shown in grey colour.}  \label{fig_sector}
\end{figure}

\section{Momentum-based Distributed Algorithm and Convergence Results}\label{sec_alg}
In this section, we propose our distributed momentum-based algorithm to solve problem~\eqref{eq_prob}-\eqref{eq_fij}. We propose a heavy-ball nonlinear perturbation-based gradient-tracking (\textbf{HBNP-GT}) dynamics as follows:
\begin{align} \label{eq_xdot_g}	
	(1-\beta)\dot{\mb{x}}_i &= -\sum_{j=1}^{n} w_{ij}^\gamma (g_l(\mb{x}_i)-g_l(\mb{x}_j))-\alpha \mb{z}_i, \\ \label{eq_ydot_g}
	\dot{\mb{z}}_i &= -\sum_{j=1}^{n} w_{ij}^\gamma (g_l(\mb{z}_i)-g_l(\mb{z}_j) ) + \partial_t \nabla f_i(\mb{x}_i),
\end{align}
where parameter $\alpha>0$ tunes the gradient-tracking (GT) and parameter $0\leq\beta<1$ adds momentum to accelerate convergence.
The auxiliary variable $\mb{z}_i$ is introduced to track the average of local gradients at neighbouring nodes. The nonlinear mapping $g_l:\mathbb{R}^{m} \mapsto \mathbb{R}^{m}$ denotes non-ideal conditions such as quantization or clipping at the links. The momentum-based nature of \textbf{HBNP-GT} can be seen by rewriting Eq.~\eqref{eq_xdot_g} in the following form,
\begin{align} \label{eq_xdot_g2}	
	\dot{\mb{x}}_i &= -\sum_{j=1}^{n} w_{ij}^\gamma (g_l(\mb{x}_i)-g_l(\mb{x}_j))-\alpha \mb{z}_i  \underbrace{+\beta \dot{\mb{x}}_i}_{\mbox{momentum}},
\end{align}
where the term $\beta \dot{\mb{x}}_i$ adds momentum to accelerate the solution towards the optimal value. The following comments are noteworthy regarding the momentum:
	\begin{itemize}
		\item The heavy-ball momentum modifies the update rule by incorporating a term proportional to the derivative of the past iterates. This accelerates convergence by using \textit{inertia} to smooth updates and to overcome local curvature issues like ill-conditioning in the objective/cost function.
		\item However, improper choice of the momentum-rate $\beta$ may cause oscillations around the stability point and, thus, divergence.
		\item The distributed nature plus GT adds coupling between $\alpha$ and $\beta$ for stability. This is discussed later via detailed convergence analysis.
	\end{itemize}

The solution is summarized in Algorithm~\ref{alg_1}. It should be noted that for different applications the data points are incorporated in the objective function $f_i(\mb{x}_i)$ and affect its gradient $\nabla f_i(\mb{x}_i)$ which is involved in the auxiliary variable dynamics \eqref{eq_ydot_g}.
\begin{algorithm} \label{alg_1}
	\textbf{Given:}  $f_{i,j}(\mb{x}_i)$, $\gamma$, $W_\gamma$, $\alpha$, $\beta$  \\	
	\textbf{Initialization:} ${\mb{z}}_i(0)=\mb{0}_{m}$
	\\
	Every Node $i$
	\For{$t<T_{end}$}{
		Finds $\boldsymbol{ \nabla} f_i(\mb{x}_i)$ \;
		Receives $\mb{x}_j$ and $\mb{z}_j$ from $j \in \mc{N}_i$ \;
		Updates $\mb{x}_i$ and $\mb{z}_i$ via Eqs.~\eqref{eq_xdot_g}-\eqref{eq_ydot_g} \;
		Shares updated states over the network $\mc{G}_\gamma$\;
	}
	\textbf{Return:}  $\mb{x}^*$ and $F^*$\;	
	\caption{\textbf{HBNP-GT} at node $i$. }
\end{algorithm}
The following lemma states the GT property of the proposed dynamics to solve the optimization problem.
\begin{lem}
	Given Assumptions~\ref{ass_net}, the state update under the proposed dynamics \eqref{eq_xdot_g}-\eqref{eq_ydot_g} for initialization $\mb{z}(0)=\mb{0}_{nm}$ satisfies the following:
\begin{align} \label{eq_sumxdot2}
	\sum_{i=1}^n\mb{z}_i = \sum_{i=1}^n \nabla f_i(\mb{x}_i),~ \sum_{i=1}^n \dot{\mb{x}}_i = -\frac{\alpha}{1-\beta} \sum_{i=1}^n \nabla f_i(\mb{x}_i).
\end{align}	
\end{lem}
\begin{proof}
Summing the dynamics \eqref{eq_xdot_g}-\eqref{eq_ydot_g} over all state nodes gives the following:
\begin{align} \label{eq_xdot_g_proof}	
	(1-\beta)\sum_{i=1}^{n}\dot{\mb{x}}_i &= -\sum_{i=1}^{n}\sum_{j=1}^{n} w_{ij}^\gamma (g_l(\mb{x}_i)-g_l(\mb{x}_j))-\alpha \sum_{i=1}^{n}\mb{z}_i, \\ \label{eq_ydot_g_proof}
	\sum_{i=1}^{n} \dot{\mb{z}}_i &= -\sum_{i=1}^{n} \sum_{j=1}^{n} w_{ij}^\gamma (g_l(\mb{z}_i)-g_l(\mb{z}_j) ) + \sum_{i=1}^{n}
	\partial_t \nabla f_i(\mb{x}_i).
\end{align}	
From the WB condition in Assumption~\ref{ass_net} and following the consensus nature of the proposed dynamics we have $\sum_{i=1}^{n}\sum_{j=1}^{n} w_{ij}^\gamma (g_l(\mb{x}_i)-g_l(\mb{x}_j))=0$ and $\sum_{i=1}^{n} \sum_{j=1}^{n} w_{ij}^\gamma (g_l(\mb{z}_i)-g_l(\mb{z}_j) )=0$. Then, from Eq.~\eqref{eq_xdot_g_proof} and~\eqref{eq_ydot_g_proof} we have
\begin{align} \label{eq_xdot_g_proof2}	
	(1-\beta)\sum_{i=1}^{n}\dot{\mb{x}}_i &=-\alpha \sum_{i=1}^{n}\mb{z}_i, \\ \label{eq_ydot_g_proof2}
	\sum_{i=1}^{n} \dot{\mb{z}}_i &= \sum_{i=1}^{n}
	\partial_t \nabla f_i(\mb{x}_i),
\end{align}
By setting initial values of~$\mb{z}(0)=\mb{0}_{nm}$ and simple integration, Eq.~\eqref{eq_ydot_g_proof2} results in $\sum_{i=1}^{n} {\mb{z}}_i = \sum_{i=1}^{n} \nabla f_i(\mb{x}_i)$. Then, from Eq.~\eqref{eq_xdot_g_proof2}, we have $(1-\beta)\sum_{i=1}^{n}\dot{\mb{x}}_i = -\alpha \sum_{i=1}^{n} \nabla f_i(\mb{x}_i)$ and the GT nature of the proposed dynamics straightly follows.
\end{proof}

%By simple integration of \eqref{eq_ydot_g}, it is clear that $(1-\beta)\sum_{i=1}^n \dot{\mb{x}}_i$ tracks $ - \sum_{i=1}^n \boldsymbol{ \nabla} f_i(\mb{x}_i)$ and the states move toward the gradient descent.  
Further, the consensus nature of dynamics \eqref{eq_xdot_g}-\eqref{eq_ydot_g} leads to the equilibrium in the form ${\mb{x}}^*= \mb{1}_n \otimes \overline{ \mb{x}}^* $, and
by setting  ${\dot {\mathbf x}_i = \mb{0}_m}$ and ${\dot {\mathbf z}_i = \mb{0}_m}$, we get
\begin{align}
	(\mathbf 1_n^\top \otimes I_m) &\boldsymbol{ \nabla} F(\mb{x}^*) = \mb{0}_m, \\
	\dot{\mb{z}}_i = \frac{d}{dt} \boldsymbol{ \nabla} f_i( \mb{x}^*)&=  \boldsymbol{ \nabla}^2 f_i( \mb{x}^*) \dot{\mb{x}}_i = \mb{0}_m,
\end{align}
which leads to
\begin{align}
	(1-\beta)\sum_{i=1}^n \dot{\mb{x}}_i = -\alpha (\mathbf 1_n^\top \otimes I_m) \boldsymbol{ \nabla} F(\mb{x}^*) =  \mb{0}_m.
\end{align}
Therefore, $[\mb{x}^*;\mb{0}_{nm}]$ is the invariant equilibrium state of the dynamics \eqref{eq_xdot_g}-\eqref{eq_ydot_g}.
For mathematical analysis and proof of stability, we rewrite the dynamics \eqref{eq_xdot_g}-\eqref{eq_ydot_g} in compact form as
\begin{align} \label{eq_xydot1}
	\left(\begin{array}{c} \dot{\mb{x}} \\ \dot{\mb{y}} \end{array} \right) = A_g(t,\alpha,\beta,\gamma) \left(\begin{array}{c} {\mb{x}} \\ {\mb{y}} \end{array} \right),
\end{align}
where system matrix $A_g(t,\alpha,\beta,\gamma)$ is defined by linearization at time $t$ as
\begin{align} \label{eq_M_g}
	\left(\begin{array}{cc} \frac{1}{1-\beta} (\overline{W}_{\gamma,g} \otimes I_m) & -\frac{\alpha}{1-\beta} I_{mn} \\ \frac{H}{1-\beta}(\overline{W}_{\gamma,g}\otimes I_m) & \overline{W}_{\gamma,g} \otimes I_m - \frac{\alpha}{1-\beta} H
	\end{array} \right),
\end{align}
where $H:=\mbox{diag}[\nabla^2 f_i(\mb{x}_i)]$. In what follows, we explain how this equation is obtained.
Define $\Xi(t) = \mbox{diag}[\xi_i(t)]$ and $\xi_i(t) = \frac{g(\mb{x}_i)}{\mb{x}_i}$ with element-wise division. From the sector-bound property in \eqref{eq_sector}, we have $\underline{\kappa} \leq \xi_i(t) \leq \overline{\kappa}$. Then, in \eqref{eq_M_g}, $\overline{W}_{\gamma,g} =  \overline{W}_{\gamma} \Xi(t)$ which follows from $g(\mb{x}(t)) = \Xi(t) \mb{x}(t)$ at every linearization point $t$.
Recall from the definition of Laplacian matrix and consensus dynamics that $-\sum_{j=1}^{n} w_{ij}^\gamma (\mb{x}_i-\mb{x}_j)=(\overline{W}_{\gamma} \otimes I_m)\mb{x}$. Then, using \eqref{eq_sector} it is straightforward to get,
$$-\sum_{j=1}^{n} w_{ij}^\gamma (g(\mb{x}_i)-g(\mb{x}_j))=(\overline{W}_{\gamma,g} \otimes I_m)\mb{x}.$$
One can derive similar equation for the auxiliary variable $\mb{z}$ by further using $\frac{d}{dt} \boldsymbol{ \nabla} f_i( \mb{x}_i)=  \boldsymbol{ \nabla}^2 f_i(\mb{x}_i) \dot{\mb{x}}_i$.

It is known that the stability of the linearized system at every linearization point proves the stability of the original nonlinear system at all times \cite{nonlin}. Therefore, for proof of stability, we need to show that Eq.~\eqref{eq_M_g} has no eigenvalues in the right-half-plane; it only has one set of $m$ eigenvalues at zero (with $m$ as the dimension of $\mb{x}$) to ensure reaching state-consensus among the agents and the rest of eigenvalues on the left-half-plane to ensure Hurwitz stability.
First, we rewrite the linearized dynamics \eqref{eq_xydot1} as
\begin{align}  \label{eq_Mg}
	A_g(t,\alpha,\beta,\gamma) &=   A_g^0 + \frac{\alpha}{1-\beta} A, \\
	\underline{\kappa} A^0 & \preceq A_g^0 \preceq \overline{\kappa} A^0,\\ \label{eq_beta_M}
	A_g^0 = \Xi(t)  A^0 &,~ \underline{\kappa} I_n  \preceq \Xi(t) \preceq \overline{\kappa} I_n,
\end{align}
where matrix $\Xi(t)=\mbox{diag}[\xi_i]$ denotes the element-wise linear model of the given nonlinearity, and
\begin{eqnarray}\nonumber
	A^0 &=&	\left(\begin{array}{cc} \frac{1}{1-\beta} (\overline{W}_{\gamma} \otimes I_m) & \mb{0}_{mn\times mn} \\ \frac{H}{1-\beta}(\overline{W}_{\gamma}\otimes I_m) & \overline{W}_{\gamma} \otimes I_m
	\end{array} \right),\\\nonumber
	A &=& \left(\begin{array}{cc} \mb{0}_{mn\times mn} & - {I_{mn}} \\ {\mb{0}_{mn\times mn}} & - H \end{array} \right).
\end{eqnarray}
Based on this formulation, one can see the term $\frac{\alpha}{1-\beta} A$ as a perturbation to the main term $A^0$ (and $A_g^0$). Moreover, due to sector-bound nonlinearity, one can derive a relation between  $\sigma(\overline{W}_{\gamma,g})$ and  $\sigma(\overline{W}_{\gamma})$, where $\sigma(\cdot)$ denotes the spectrum of a matrix.
This is because $\Xi(t)$ (which models the nonlinearity at every linearization point) is a diagonal matrix and, thus, $\mbox{det}(\overline{W}_{\gamma,g}-\lambda I_{2mn}) = \mbox{det}(\overline{W}_{\gamma}-\lambda \Xi(t)^{-1} I_{2mn})=0$.
In other words, at every linearization point, given that $\lambda_i \in \sigma(\overline{W}_{\gamma})$, it is straightforward to see that $\lambda_i \xi_i \in \sigma(\overline{W}_{\gamma,g})$.

Given the above preliminaries, in the rest of this section, we use perturbation theory at every linearization point along with Lyapunov theory to prove system stability. For this, first, recall the following lemma.
\begin{lem} \label{lem_dM}
	\cite{stewart_book,cai2012average} Consider an $n$-by-$n$ parameter-dependent matrix~$P(\alpha)$ as a function of~${\alpha \geq 0}$. Given that $P$ has~${l<n}$ repetetive eigenvalues~$\lambda_1=\ldots=\lambda_l$ with (right and left) linearly independent unit eigenvectors~$\mb{v}_1,\ldots,\mb{v}_l$ and~$\mb{u}_1,\ldots,\mb{u}_l$, let $\lambda_i(\alpha)$ be the eigenvalue of~$P(\alpha)$. Then, $\frac{d\lambda_i}{d\alpha}|_{\alpha=0}$ is the $i$-th eigenvalue of
	\begin{align} \label{eq_dalpha}
		\left(\begin{array}{ccc}
			\mb{u}_1^\top P' \mb{v}_1 & \ldots & \mb{u}_1^\top P' \mb{v}_l \\
			& \ddots & \\
			\mb{u}_l^\top P' \mb{v}_1 & \ldots & \mb{u}_l^\top P' \mb{v}_l
		\end{array} \right), ~ P' =  \frac{dP(\alpha)}{d\alpha}|_{\alpha=0}.
	\end{align}
\end{lem}
\begin{lem} \label{lem_sigma}
	\cite[Appendix]{delay_est}
	Given a block-structured matrix $A$ in the form
	\begin{align} \label{eq_block}
		M = \left(\begin{array}{cc}
			D & E \\
			F & G
		\end{array} \right),
	\end{align}	
	and having $\mbox{det}(E)\neq 0$, determinant of $A$ satisfies
	\begin{align} \label{eq_block0}
		\mbox{det}(M) =	\mbox{det}(E) \mbox{det}(F-DE^{-1}G).
	\end{align}	
\end{lem}
These lemmas are used in the following theorem to prove the main result on convergence and stability.
\begin{theorem} \label{thm_zeroeig}
	Let Assumptions~\ref{ass_nonconv}-\ref{ass_nonlin} hold. For $\beta \leq 1-\sqrt{\frac{{\alpha}\zeta}{|\operatorname{Re}\{\lambda_2\}|}}$ and $\alpha \leq \frac{|\operatorname{Re}\{\lambda_2\}|(1-\beta)^2}{\zeta}$  with $\lambda_2$ as the largest non-zero eigenvalue of $A_g$, the compact system dynamics~\eqref{eq_xydot1}-\eqref{eq_M_g} has $m$ zero eigenvalues $\forall t>0$  with the rest of its eigenvalues in the left-half-plane. Further, the proposed dynamics converges to the optimal point $[\mb{x}^*;\mb{0}_{nm}]$.
\end{theorem}
\begin{proof}
	The proof analysis is divided into three steps. First, in Step I, we find the eigen-spectrum of $A_g^0$. Then, in Step II, we discuss how the eigen-spectrum of the matrix $A_g$ is affected as the perturbed version of $A_g^0$ for stability analysis. Finally, in Step III, we prove convergence by residual analysis and using the Lyapunov stability theorem.\\	
	\textbf{Step I:}  Recall that, for the proposed nonlinear dynamics~\eqref{eq_xdot_g}-\eqref{eq_ydot_g}, the eigenvalues of $A^0_g$ and $A^0$ at every linearization point are related by
	\begin{align} \label{eq_spect_k}
		\underline{\kappa} \sigma(A^0) \leq \sigma(A^0_g) \leq \overline{\kappa} \sigma(A^0).
	\end{align}
	Since $A^0$ is block triangular, we have $\sigma(A^0) = \sigma(\frac{1}{1-\beta} (\overline{W}_{\gamma} \otimes I_m)) \cup \sigma(\overline{W}_\gamma \otimes I_m)$.	
	Following the definition of the Laplacian matrix, $\sigma(\overline{W}_\gamma \otimes I_m)$ includes a set of~$m$ zero eigenvalues and the rest in the left-half-plane. Therefore, the eigenvalues of the matrix $A^0$ satisfy
	$$\operatorname{Re}\{\lambda_{2n,j}\} \leq \ldots \leq \operatorname{Re}\{\lambda_{3,j}\} < \lambda_{2,j} = \lambda_{1,j} = 0,$$
	for $j=\{1,\ldots,m\}$.\\	
	\textbf{Step II:} From Eq.~\eqref{eq_Mg}, one can see $\sigma(A_g)$ as the perturbed version of $\sigma(A_g^0)$ by $\sigma(A^0)$ via parameter $\frac{\alpha}{1-\beta}$. based on this, we check how the eigenvalues of $A_g^0$, in particular zero eigenvalues $\lambda_{1,j}$ and~$\lambda_{2,j}$, are affected by the perturbation $\frac{\alpha}{1-\beta}A$. This is done by using Lemma~\ref{lem_dM}. Noting that $\Xi(t)$ is a diagonal matrix, define the right eigenvectors of~$\lambda_{1,j}$ and~$\lambda_{2,j}$ as follows,
	\begin{align} \nonumber
		V = [V_1~V_2] =\left(\begin{array}{cc}
			\mb{1}_n& \mb{0}_n \\
			\mb{0}_n & \mb{1}_n
		\end{array} \right)\otimes I_m,
	\end{align}
	and the left eigenvectors as~$V^\top$. Let define $\eta:=\frac{\alpha}{1-\beta}$ and denote the perturbed eigenvalues as $\lambda_{1,j}(\eta,t)$ and~$\lambda_{2,j}(\eta,t)$. Using Eqs. \eqref{eq_M_g} and \eqref{eq_dalpha}, we have~$\frac{dA_g(\eta)}{d\eta}|_{\eta=0}=A$ and	
	%\footnote{If the zero eigenvalues move to the left-half-plane the system is stable and if they move toward the right-half-plane the system dynamics is unstable.}
	\begin{eqnarray} \label{eq_dmalpha}
		V^\top A V= \left(\begin{array}{cc}
			\mb{0}_{m\times m}	& \mb{0}_{m\times m} \\
			...	& -(\mb{1}_n \otimes I_m)^\top H (\mb{1}_n \otimes I_m)
		\end{array} \right).
	\end{eqnarray}
	This is a block-triangular matrix with $m$ eigenvalues at zero.
	The sign of other $m$  eigenvalues is defined based on the following inequality,
	\begin{equation} \label{eq_sum_df}
		-(\mb{1}_n \otimes I_m)^\top H  ( \mb{1}_n \otimes I_m)= -\sum_{i=1}^n \nabla^2  f_i(\mb{x}_i) \prec 0,
	\end{equation}
	From Eq.~\ref{eq_sum_df} and Lemma~\ref{lem_dM} we have ${\frac{d\lambda_{1,j}}{d\eta}|_{\eta=0} = 0}$ and~${\frac{d\lambda_{2,j}}{d\eta}}|_{\eta=0}<0$. This implies that by perturbing $A^0_g$ by~$\eta A$, one set of $m$ zero eigenvalues remains at zero and one set of $m$ zero eigenvalues moves to the left-half-plane and toward stability.
	\\
	Next, we check that other eigenvalues of $A_g^0$ do not move to the right half-plane due to perturbation. For notation simplicity and without loss of generality, we set $m=1$ and drop the dependence on $\gamma,g, t$.  Then, one can find $\sigma(A_g)$  from Lemma~\ref{lem_sigma} as
	\begin{align} \nonumber
		\mbox{det}(\frac{\alpha}{1-\beta}  I_{n})& \mbox{det}\Big(\frac{H\overline{W}}{1-\beta} +\\ &\frac{1-\beta}{\alpha}((\frac{\overline{W}}{1-\beta}  -\lambda I_{n})(\overline{W}  -\lambda I_{n}- \frac{\alpha H}{1-\beta}) \Big) = 0.
	\end{align}
	We have $\mbox{det}(\frac{\alpha}{1-\beta}  I_{n})\neq 0$. Then, after some simplification and factorization, we get
	\begin{align} \nonumber
		\mbox{det}\Big(\frac{H\overline{W}}{1-\beta} +\frac{1}{\alpha}((&\overline{W}  -\lambda(1-\beta) I_{n})(\overline{W}  -\lambda I_{n})\\ &- \frac{\alpha H}{1-\beta}(\overline{W}  -\lambda(1-\beta) I_{n}) \Big) = 0,
	\end{align}
	which results in
	\begin{align} \nonumber
		&\mbox{det}\Big((\overline{W}  -\lambda(1-\beta) I_{n})(\overline{W}  -\lambda I_{n})+ \frac{\lambda \alpha H}{1-\beta} \Big) = \\ \nonumber
		& \mbox{det}\Big((\overline{W}  -\lambda I_{n})^2+ \lambda \beta (\overline{W}  -\lambda I_{n})+ \frac{\lambda \alpha H}{1-\beta} \Big) = \\ \nonumber
		& \mbox{det}\Big((\overline{W}  - \lambda(1-\frac{\beta}{2}) I_{n})^2 - \frac{\lambda^2 \beta^2}{4}I_{n} + \frac{\lambda \alpha H}{1-\beta} \Big) =\\ \label{eq_proof_pertb}
		& \mbox{det}\Big((\overline{W}  - \lambda(1-\frac{\beta}{2}) (I_{n} \pm \underbrace{\sqrt{\frac{\frac{\beta^2}{4} - \frac{ \alpha H}{(1-\beta)\lambda}}{(1-\frac{\beta}{2})^2}}}_{\mbox{perturbation term}}) \Big)	=0.
	\end{align}
	First, note that for $\lambda=0$ all the above are zero, and $\lambda=0$ is one eigenvalue of the system.
	Recall that all the eigenvalues of the Laplacian matrix $\overline{W}$ are in the left-half-plane except one isolated zero eigenvalue \cite{olfatisaberfaxmurray07}. For $\alpha=0$ and $\beta=0$ the eigen-spectrum is the same as $\sigma(\overline{W})$. For large positive values of $\alpha,\beta$, the perturbation term in \eqref{eq_proof_pertb} may lead the stable eigenvalues toward instability. Knowing that the system eigenvalues are continuous functions of the matrix entries \cite{stewart_book}, we need to find the admissible range for $\alpha$ (and $\beta$) such that the perturbation term in \eqref{eq_proof_pertb} does not make the system unstable, i.e., to find the range $\alpha \in (0,\overline{\alpha})$ such that the non-zero eigenvalues remain in the left-half-plane. Therefore, for $\lambda \neq 0$, we need to solve the following for given $\beta \in [0,1)$:
	\begin{align}
		\overline{\alpha} = \argmin_{\alpha} \left|1 - \sqrt{\frac{\frac{\beta^2}{4} - \frac{ \alpha H}{(1-\beta)\lambda}}{(1-\frac{\beta}{2})^2}}\right|.
	\end{align}
	The minimum value of the above is zero, for which the perturbation term needs to be equal to one. Thus, we set
	\begin{align} \nonumber
		&(1-\frac{\beta}{2})^2-\frac{\beta^2}{4} = \frac{ \overline{\alpha} H}{(1-\beta)|\lambda|}, \\
		&(1-\beta) = \frac{ \overline{\alpha} H}{(1-\beta)|\lambda|}.	
	\end{align}
	Recall from Assumption~\ref{ass_nonconv} that $H \preceq \zeta I_{n}$. Also, for $\overline{W}$ we have $|\operatorname{Re}\{\lambda_2\}| = \min_{i=1,\dots,n} |\operatorname{Re}\{\lambda_i\}|$ for non-zero eigenvalues.
	Then, the maximum admissible value of $\overline{\alpha}$ for stability satisfies,
	\begin{align} \label{eq_alphabar0}
		\overline{\alpha} \leq \frac{|\operatorname{Re}\{\lambda_2\}|(1-\beta)^2}{\zeta}.
	\end{align}
	The above implies that for $\alpha<\overline{\alpha}$, the nonzero eigenvalues remain in the left half-plane, and the proof of system stability follows. As a follow-up, for a given $\alpha$, the admissible range for $\beta$ can be defined as follows,
	\begin{align} \label{eq_beta0}
		\beta \leq 1-\sqrt{\frac{{\alpha}\zeta}{|\operatorname{Re}\{\lambda_2\}|}}.
	\end{align}
	\textbf{Step III:} Define the following residual vector for Lyapunov stability analysis:
	\begin{align} \label{eq_delta}
		{\delta} = \left(\begin{array}{c} {\mb{x}} \\ {\mb{y}} \end{array} \right) - \left(\begin{array}{c} {\mb{x}^*} \\ {\mb{0}_{nm}} \end{array} \right),
	\end{align}
	and the following Lyapunov function,
	\begin{align} \label{eq_V}
		\mc{V}(\delta) = \frac{1}{2} \delta^\top \delta =  \frac{1}{2}\lVert \delta \rVert_2^2.
	\end{align}
	This function takes positive values and for $\delta \rightarrow \mb{0}_{2nm}$ we have $\mc{V} \rightarrow 0$ (i.e., a  positive-semi-definite function). Applying the proposed dynamics \eqref{eq_xydot1}-\eqref{eq_M_g}, we have
	\begin{align} \nonumber
		\dot{\delta} = A_g \left(\begin{array}{c} \mb{x} \\ \mb{y} \end{array} \right) - A_g\left(\begin{array}{c} \mb{x}^* \\ \mb{0}_{mn} \end{array} \right) = A_g \delta.
	\end{align}
	Therefore, $\dot{\mc{V}} = {\delta}^\top \dot{\delta}=  \delta^\top A_g {\delta}$. Using the results  from~\cite[Sections~VIII-IX]{SensNets:Olfati04}, we get
	\begin{eqnarray} \label{eq_Re2}
		\dot{\mc{V}} = \delta^\top A_g {\delta} \leq \operatorname{Re}\{{\lambda}_{2}\} \delta^\top  \delta,
	\end{eqnarray}
	Recall that ${\lambda}_{2}$ is in the left-half-plane, which implies that $\dot{\mc{V}}$ is negative-definite for $\delta \neq \mb{0}_{2nm}$. Then, from Lyapunov stability theorem \cite{nonlin} the residual $\delta$ under the proposed dynamics converges to $\mb{0}_{2nm}$, i.e., we have $\mb{x} \rightarrow \mb{x}^* $ and $\mb{y} \rightarrow \mb{0}_{nm}$.
	This completes the proof.
\end{proof}

\begin{rem}
Recall that for the algorithm convergence and stability of the solution, from Eq.~\eqref{eq_alphabar0} we need ${\alpha} \leq \frac{|\operatorname{Re}\{\lambda_2\}|(1-\beta)^2}{\zeta}$ or equivalently $\beta \leq 1-\sqrt{\frac{{\alpha}\zeta}{|\operatorname{Re}\{\lambda_2\}|}}$ from Eq. \eqref{eq_beta0}. It is clear that there is a trade-off between $\alpha$ as the GT rate and $\beta$ as the momentum rate. For large values of momentum ($\beta \rightarrow 1$), the bound on $\alpha$ is tighter (${\alpha} \rightarrow 0$). Therefore, one needs to balance between the momentum term and GT term in dynamics \eqref{eq_xdot_g}.	
\end{rem}

%\section{Some Discussions} \label{sec_disc}

\begin{rem}
	In the proof of Theorem~\ref{thm_zeroeig}, we only used Assumption~\ref{ass_nonconv} on the convexity of local costs. Eq.~\eqref{eq_ass1} is used in \eqref{eq_alphabar0} and Eq.~\eqref{eq_ass2} is used in \eqref{eq_sum_df}. Therefore, Theorem~\ref{thm_zeroeig} converges even for non-convex local costs $f_i(\mb{x}_i)$ satisfying Assumption~\ref{ass_nonconv}. A numerical example is discussed in the next section.
\end{rem}

	Following Eq.~\eqref{eq_alphabar0} and~\eqref{eq_beta0}, the algebraic connectivity (or Fiedler eigenvalue) of the network, denoted by $\lambda_2$, plays a key role in the gradient tracking and admissible momentum rate. Larger $\lambda_2$ allows for larger $\alpha$ and $\beta$ rates. If the network topology is not locally known, one can use approximation methods to evaluate $\lambda_2$; for example, for undirected networks, we have $\lambda_2 \geq \frac{1}{nd}$  with $d$ as the network diameter. It is known that densely connected networks have larger algebraic connectivity \cite{SensNets:Olfati04}; this implies a larger admissible range of $\alpha$ and $\beta$ rates for such networks.

\begin{rem} \label{rem_struc}
	The other factor impacting the convergence rate and optimality gap is the network topology. In \cite{nedic2018network,ying2021exponential} it is claimed that structured networks such as exponential graph topologies show faster convergence as compared to general unstructured ad-hoc networks or random networks. In the next section, it is explicitly shown that for networks of the same size, structured exponential graphs reach faster convergence than unstructured random Erdos-Renyi (ER) graph topologies.
\end{rem}

	In the proof of Theorem~\ref{thm_zeroeig}, using perturbation theory for system stability allows to address switching networks. To be more specific, we only need that the perturbed eigenvalues remain in the left-half plane at all times, irrespective of the change in the network topology, and the proof is irrespective of the time-varying structure of the network as long as it is strongly-connected and WB at all times. Recall from Fig.~\ref{fig_remov} that the WB condition also allows for handling link failure (e.g., due to packet drops or congestion) in contrast to the existing weight-stochastic requirement in the literature, which requires redesigning the stochastic weights.

\section{Numerical Simulations} \label{sec_sim}
\subsection{Locally Non-convex Cost}
First, we consider locally non-convex cost functions satisfying Assumption~\ref{ass_nonconv} as follows:
\begin{align}\label{eq_fij_sim}
	f_{i,j}(x_i) = 2 x_i^2 +\cos^2(x_i)+a_{i,j} \sin(x_i) + b_{i,j}x_i,
\end{align}
with $a_{i,j}$ and $b_{i,j}$ set randomly in $[-5,5]$ such that $\sum_{i=1}^n \sum_{j=1}^m a_{i,j} = 0$, $\sum_{i=1}^n \sum_{j=1}^m b_{i,j}=0$, and $a_{i,j},b_{i,j} \neq 0$. For this locally non-convex cost function, we have
\begin{align}\label{eq_dfij_sim}
	\nabla^2 f_{i,j}(x_i) = 4 -2\cos(2x_i)-a_{i,j} \sin(x_i), 
\end{align}
for which $\nabla^2 f_{i,j}(x_i) \leq 11$  and $\sum_{i=1}^n \nabla^2 f_{i,j}(x_i) = 4n -2\sum_{i=1}^n \cos(2x_i) > 0$ and, thus, satisfies Assumption~\ref{ass_nonconv}.  
Samples of these locally non-convex costs and the global cost as $F(\mb x) = \sum_{i=1}^{n} \sum_{j=1}^{m} f_{i,j}(\mb{x}_i) $ are shown in Fig.~\ref{fig_nonconv}.
\begin{figure*}
	\centering
	\includegraphics[width=1.75in]{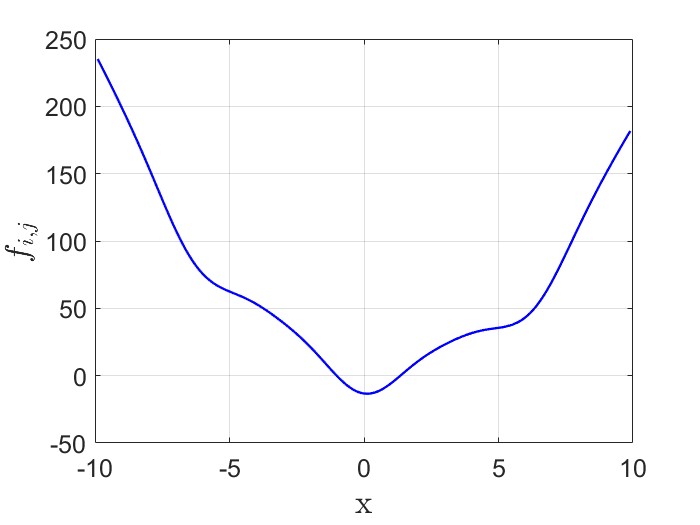}
	\includegraphics[width=1.75in]{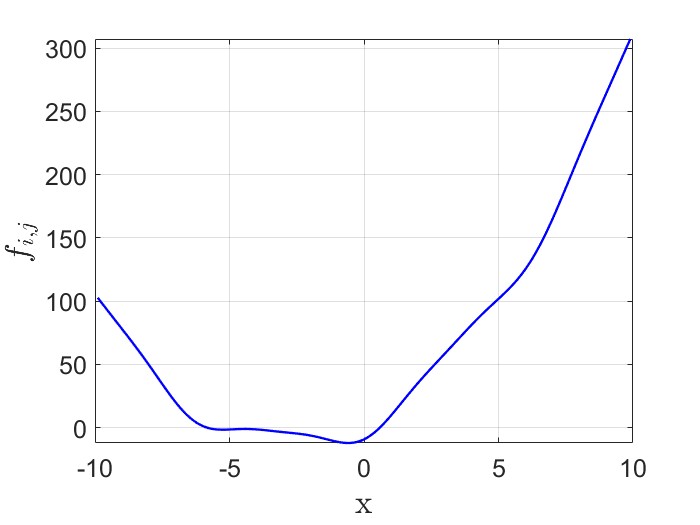}
	\includegraphics[width=1.75in]{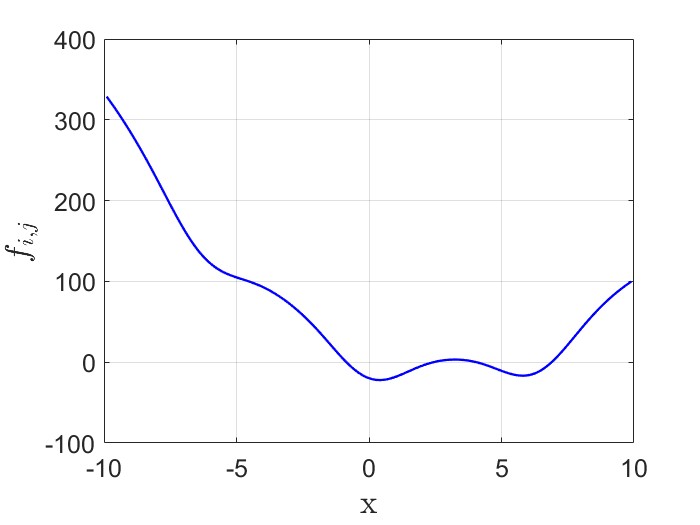}
	\includegraphics[width=1.75in]{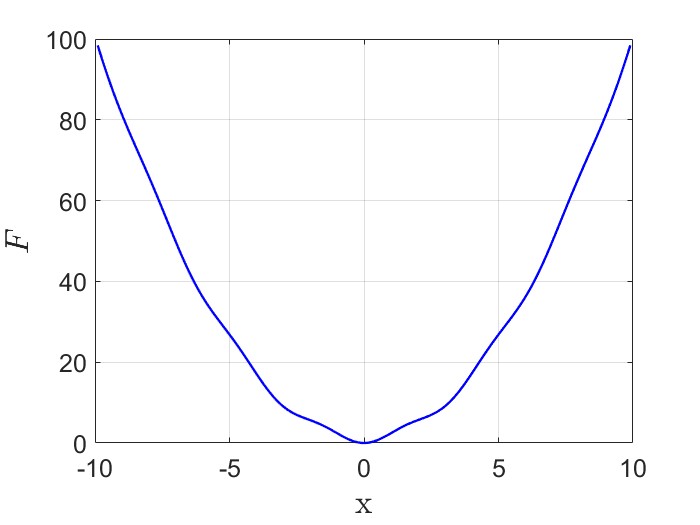}	 		
	\caption{The three figures on the top show the locally non-convex cost function $f_{i,j}(\mb{x}_i)$ in \eqref{eq_fij_sim} at three sample computing nodes. The figure on the bottom shows the global cost function $F(\mb x)$. }  \label{fig_nonconv}
\end{figure*}
The parameters of \textbf{HBNP-GT} algorithm are set as $\alpha=1$ for the GT rate, $\beta=0.6$ for the momentum rate, and $n=10$ as the number of computing nodes. For this example, we consider log-quantization at the links defined as
\begin{align}\label{eq_hl_qlog}
	h_l(z) = \mbox{sgn}(z)\exp\left(\rho\left[\dfrac{\log(|z|)}{\rho}\right] \right),
\end{align}
with $[\cdot]$ as rounding to the nearest integer and $\mbox{sgn}(\cdot)$ as the sign function. For log-scale quantization \eqref{eq_hl_qlog} we have 
	\begin{align} \label{eq_sectorq}
		\underline{\kappa} z = (1-\frac{\rho}{2})z  \leq h_l(z) \leq (1+\frac{\rho}{2})z = \overline{\kappa} z,
\end{align}
which satisfies Assumption~\ref{ass_nonlin}.
For this simulation we set $\rho = \frac{1}{64}$. Two networks are considered: a structured exponential network and an unstructured random ER network. The link weights are set randomly such that they satisfy the WB condition in Assumption~\ref{ass_net}. To solve the distributed optimization problem, we use Algorithm~\ref{alg_1} and the proposed dynamics~\eqref{eq_xdot_g}-\eqref{eq_ydot_g}. While our algorithm is presented in a continuous-time framework for analytical clarity and to apply tools from continuous-time dynamical systems theory, for practical implementation, we use a discrete-time approximation using the Euler approximation technique. The simulation is shown in Fig.~\ref{fig_nonconv_sim}. As stated in Remark~\ref{rem_struc}, the structured network of the same size and connectivity reaches faster convergence with a low optimality gap. It is clear from the figure that the solution converges in the presence of logarithmically quantized information exchange, and the rate of convergence can be improved by increasing the momentum parameter $\beta$ (within the admissible range).
\begin{figure}
	\centering
	\includegraphics[width=2.25in]{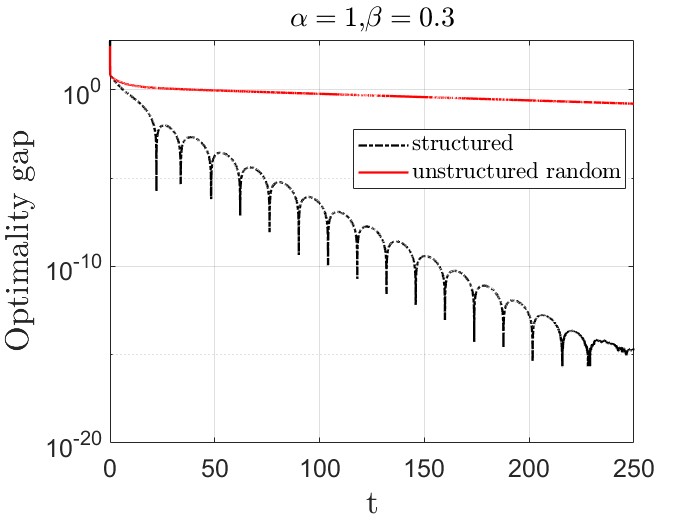}
	\includegraphics[width=2.25in]{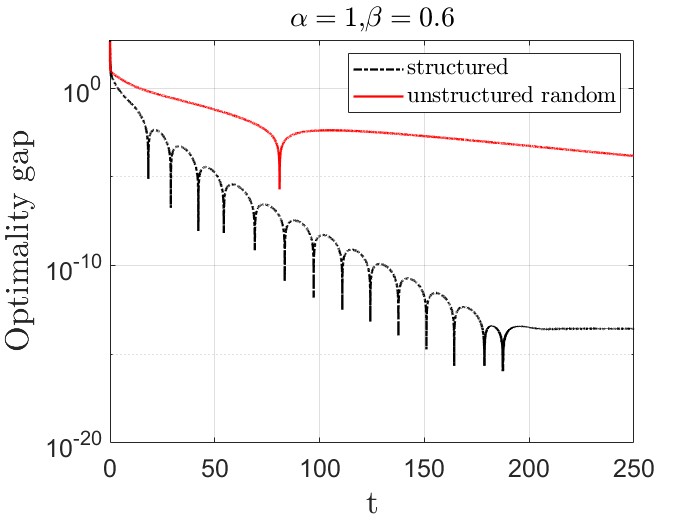}	 		 		
	\caption{The convergence of the \textbf{HBNP-GT} algorithm under logarithmically quantized data exchange over structured and unstructured networks with $\beta=0.3$ (Left) and $\beta=0.6$ (Right). Clearly, structured networks lead to faster convergence. }  \label{fig_nonconv_sim}
\end{figure}

Next, we redo the simulation for different values of $\beta$ and $\alpha$ by fixing the value of one and changing the other one. This illustrates how the change in one of the parameters $\alpha$ or $\beta$ affects the convergence rate. Fig.~\ref{fig_nonconv_sim2} presents the simulation results. As it is clear from the figure, increasing either of the $\alpha$ or $\beta$ parameters accelerates the convergence. However, the rate of $\alpha$ and $\beta$ must satisfy Eq.~\eqref{eq_alphabar0} and \eqref{eq_beta0}.
\begin{figure}
	\centering
	\includegraphics[width=2.25in]{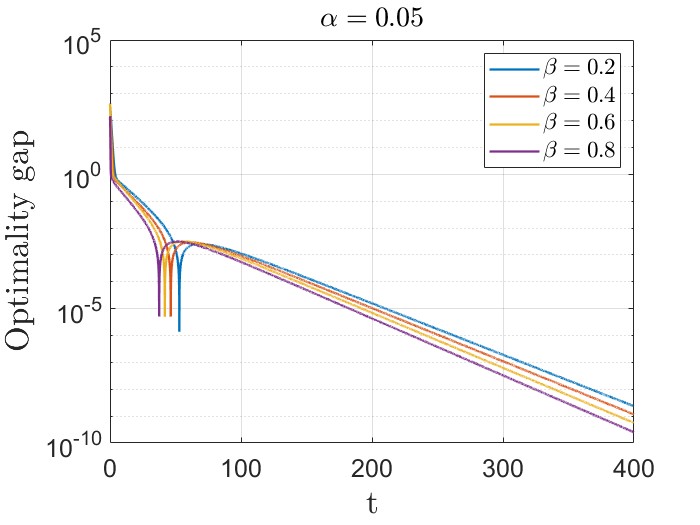}
	\includegraphics[width=2.25in]{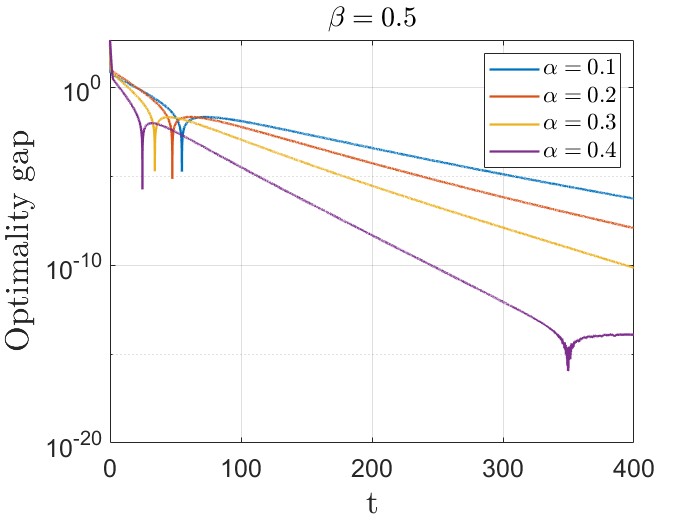}	 		 		
	\caption{The convergence rate of the \textbf{HBNP-GT} algorithm for different values of $\alpha$ or $\beta$: (Left) changing $\beta$ and setting fixed-value $\alpha=0.05$ (Right) changing $\alpha$ and setting fixed-value $\beta=0.5$.  }  \label{fig_nonconv_sim2}
\end{figure}

\subsection{Distributed SVM}
Next, we consider a set of $m=100$ data points distributed between a group of $n=6$ computing nodes/agents. The network is considered as a simple cycle with link weights equal to $\frac{1}{2}$ which satisfies the WB condition in Assumption~\ref{ass_net}. Every node has access to $75\%$ of the overall data points (randomly chosen). The goal is to classify these data points locally using a distributed support-vector-machine (SVM). The data points are not separable linearly via the SVM hyperplane and, thus, a Kernel mapping is used to transform the data points into a linearly separable form as shown in Fig.~\ref{fig_data}.
\begin{figure}
	\centering
	\includegraphics[width=2.25in]{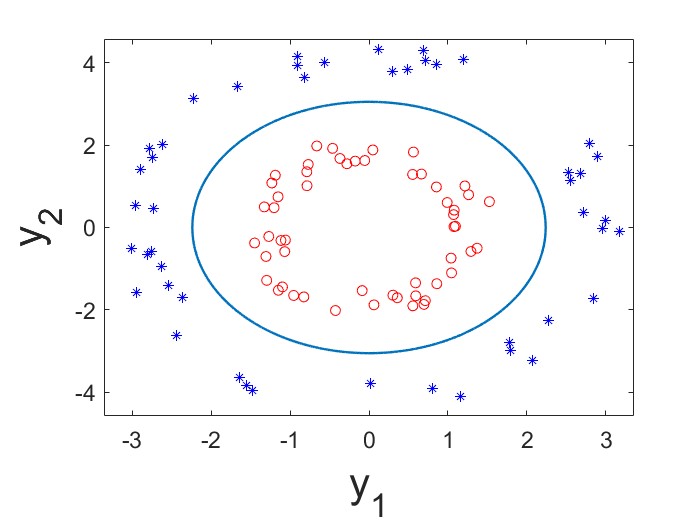}
	\includegraphics[width=2.25in]{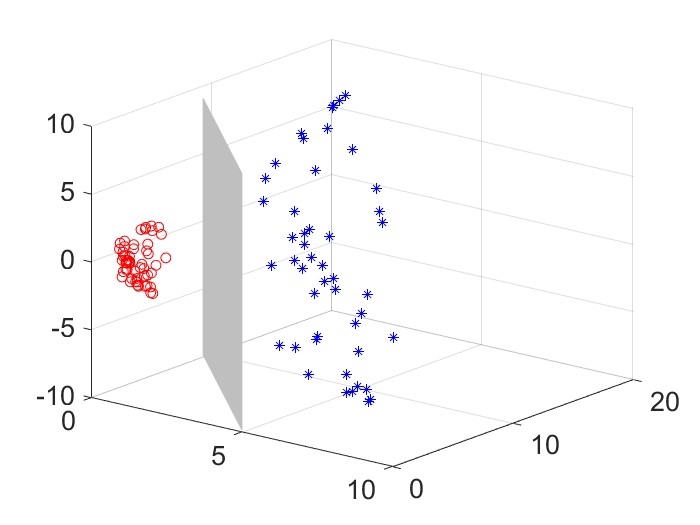}	 		 		
	\caption{(Left) The data points given in 2D space are not linearly seperable. (Right) The data points are transformed into 3D space with proper Kernel mapping and are linearly separable via the SVM hyperplane in this new form. }  \label{fig_data}
\end{figure}
The following cost function (known as the logarithmic hinge loss) is used to classify the data points:
\begin{align}\label{eq_svm_smooth}
	&f_i(\mb{x}_i)=\boldsymbol{\omega}_i^\top \boldsymbol{\omega}_i + C \sum_{j=1}^{m_i} \tfrac{1}{\mu}\log (1+\exp(\mu y_i)),\\
	&{\mb{x}_i = [\boldsymbol{\omega}_i^\top;\nu_i]}\in\mathbb{R}^4,~~{y_i=1-l_j( \boldsymbol{\omega}_i^\top \phi(\boldsymbol{\chi}^i_j)-\nu_i)},
\end{align}
where $\mb{x}_i=[\boldsymbol{\omega}_i;\nu_i]$ denotes the SVM classifier parameters at node $i$, $\boldsymbol{\chi}^i_j$ denotes the coordinates of the data point $j$ given to node $i$,  $l_j$ is the label of this data point, and $\phi(\cdot)$ is the Kernel mapping. The objective function is convex \cite{slp_book} and for finite number of bounded-value data points satisfies Assumption~\ref{ass_nonconv}. The parameters of the hing loss function are set as $\mu=3$ and $C=2$, and the parameter of the \textbf{HBNP-GT} algorithm is set as $\alpha=6$ and $\beta=0.5$ with $\rho = \frac{1}{128}$ as the level of log-scale quantization. Recall from the previous subsection that log-scale quantization is sector-bound nonlinear and satisfies Assumption~\ref{ass_nonlin}. The time evolution of the hyperline parameters is shown in Fig.~\ref{fig_svm}. The figure clearly shows that the hyperplane parameters $\omega_i,\nu_i$ converge to the optimal value.
\begin{figure}
	\centering
	\includegraphics[width=2.25in]{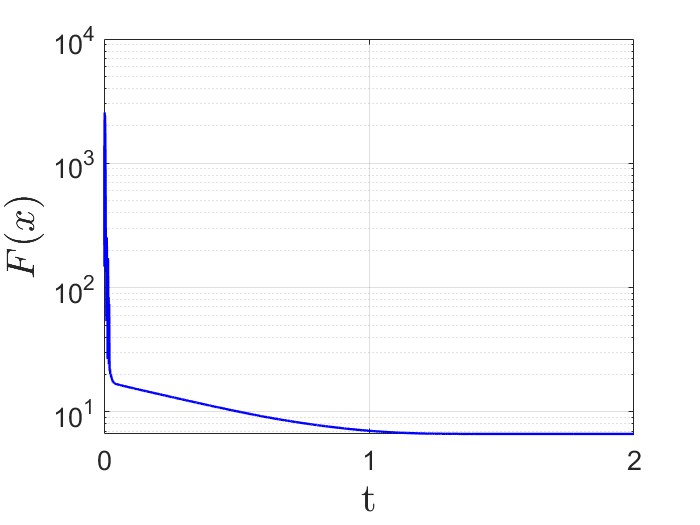}
	\includegraphics[width=2.25in]{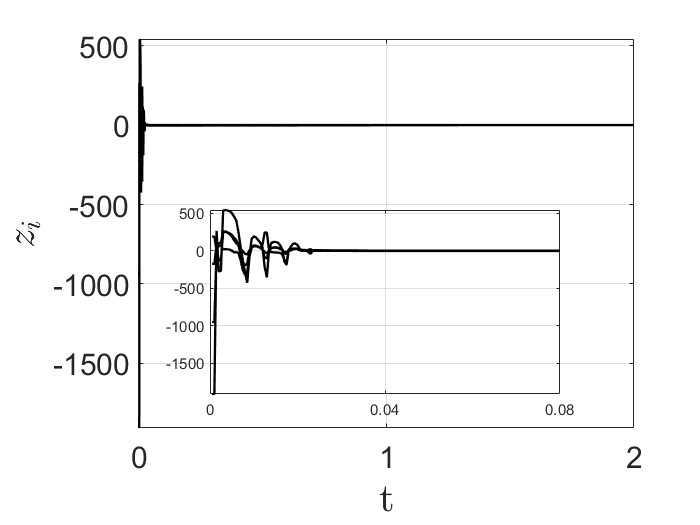}
	\includegraphics[width=2.25in]{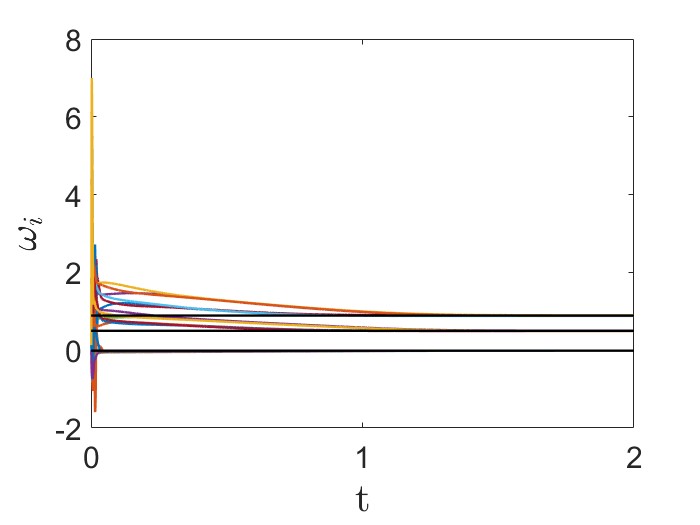}
	\includegraphics[width=2.25in]{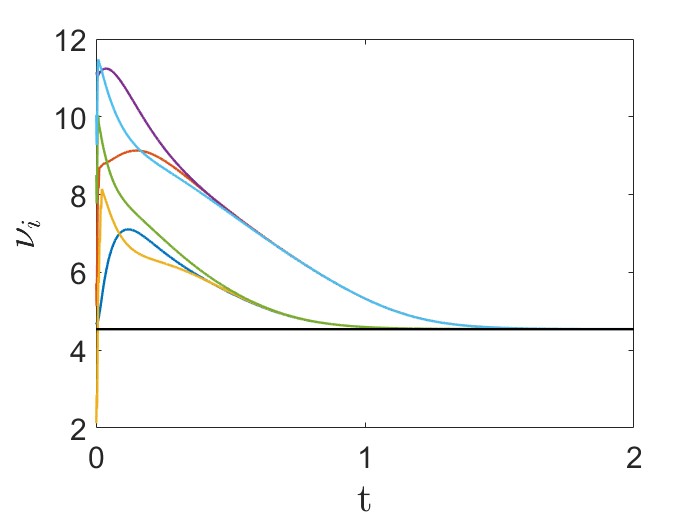}		 		 		
	\caption{The time evolution of the SVM parameters for the data-set given in Fig.~\ref{fig_data}. (TopLeft) The loss function decreases over time. (TopRight) The auxiliary variable $\mb{z}_i$ (which tracks the sum of the gradients) at all nodes converges to zero. (Bottom) The hyperplane parameters $\omega_i,\nu_i$ converge to the optimal value denoted by the black line. }  \label{fig_svm}
\end{figure}

\subsection{Distributed Linear Regression}
Consider a set of data points $[\boldsymbol{\chi}_i,y_i] \in \mathbb{R}^2$ shown in Fig.~\ref{fig_lr}-TopLeft. These data are distributed among $n=15$ nodes to locally find the regressor line in the form $ \boldsymbol{\beta}^\top \boldsymbol{\chi}_i - \nu = y_i$. \
The objective function for this problem is given as,
\begin{align} \label{eq_lr_cost}
	&f_i(\mb{x}_i) = \sum_{j=1}^{m_i} (\boldsymbol{\beta}_i^\top \boldsymbol{\chi}^i_j - \nu_i -y_j)^2,\\
	&{\mb{x}_i = [\boldsymbol{\beta}_i;\nu_i]}\in\mathbb{R}^2,
\end{align}
where $\mb{x}_i=[\boldsymbol{\beta}_i;\nu_i]$ is the regressor line parameters at node $i$, $[\boldsymbol{\chi}^i_j,y_i]$ is the coordinates of the data point $j$ at node $i$. This objective function is quadratic and therefore satisfies Assumption~\ref{ass_nonconv}.
For this example, we consider clipping at the links defined as
\begin{align}\label{eq_hl_sat}
	h_l(z) = \left\{ \begin{array}{ll}
		z, & \text{if}~  -\rho\leq z\leq \rho,\\
		\rho, & \text{otherwise}.
	\end{array}\right.
\end{align}
with the clipping level $\rho = 10$. Given that $z_{\min} \leq z \leq z_{\max}$, the sector bounds for this nonlinear mapping are defined as $ \frac{\rho}{\max\{|z_{\max}|,|z_{\min}|\}}z  \leq h_l(z) \leq z $ (with $\underline{\kappa} = \frac{\rho}{\max\{|z_{\max}|,|z_{\min}|\}}$ and $ \overline{\kappa}=1$) satisfying Assumption~\ref{ass_nonlin}. The network of agents is a time-varying random ER graph with $40\%$ linking probability, where the topology changes every $0.0015$ sec. The link weights of ER graphs are designed to be WB satisfying Assumption~\ref{ass_net}. The time-evolution of the regression cost and the regression parameters at all nodes under the \textbf{HBNP-GT} algorithm (with $\alpha=3$ and $\beta=0.4$) are shown in Fig.~\ref{fig_lr}. Clearly, the cost function decreases over time, and the regressor line parameters at all computing nodes reach consensus on the optimal value.

\begin{figure}
	\centering
	\includegraphics[width=2.25in]{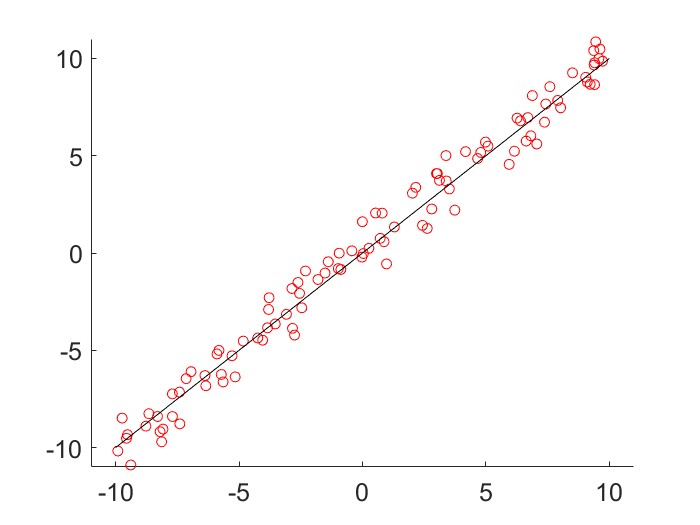}
	\includegraphics[width=2.25in]{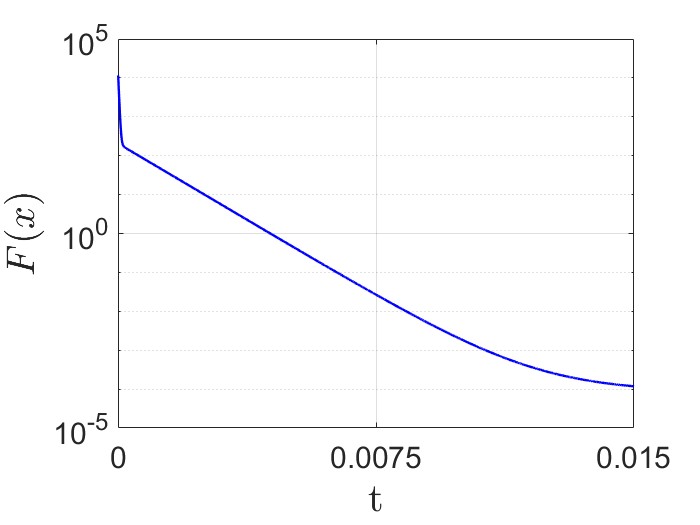}
	\includegraphics[width=2.25in]{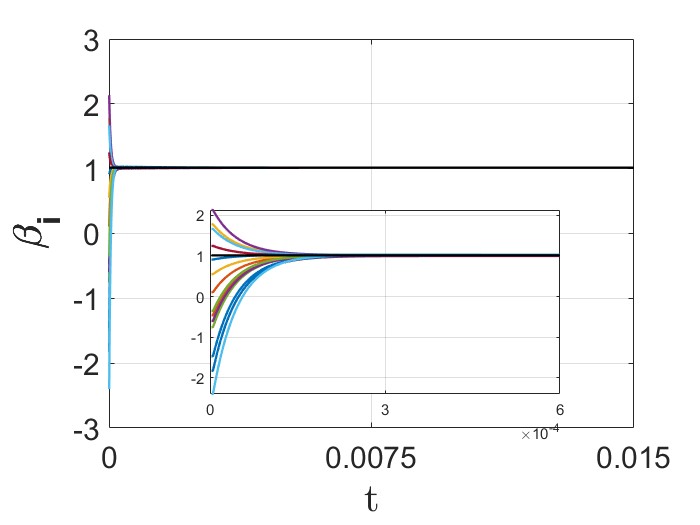}
	\includegraphics[width=2.25in]{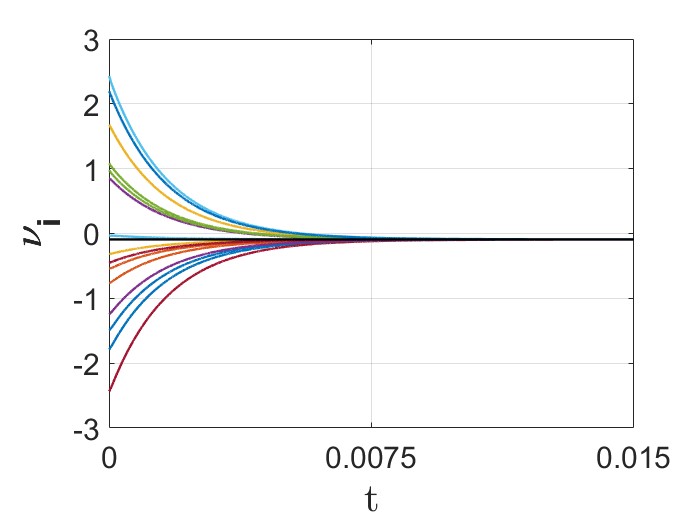}		 		 		
	\caption{The data points and the regressor line parameters versus time under \textbf{HBNP-GT} algorithm. (TopLeft) The data points and the optimal regressor line using centralized optimization. (TopRight) The regression cost function decreases over time towards the optimal value. (Bottom) The regressor line parameters $\boldsymbol{\beta}_i,\nu_i$ reach consensus and converge to the optimal value. }  \label{fig_lr}
\end{figure}

\subsection{Distributed Logistic Regression: Comparison with the Literature}
	In this section, we consider a real data-set example and compare our solution with the existing literature. For the simulation, $N = 12000$ labelled image data are randomly selected from the MNIST dataset. The purpose is to classify these images using logistic regression with a convex regularizer, where the data is distributed among $n=16$ nodes. The objective function is defined as 
\begin{align}
	\min_{\mb{b},c} &
	F(\mb{b},c) = \frac{1}{n}\sum_{i=1}^{n} f_i,
\end{align}  
with every node $i$ accessing a batch of $m_i=\frac{N}{n}=750$ sample images. Then, every node locally minimizes the following training objective:
\begin{align}\label{eq_fij_regression}
	f_i(\mb{x}) = \frac{1}{m_i}\sum_{j=1}^{m_i} \ln(1+\exp(-(\mb{b}^\top \mc{X}_{i,j}+c)\mc{Y}_{i,j}))+\frac{\theta}{2}\|\mb{b}\|_2^2,
\end{align}
which is smooth because of the addition of the regularizer $\theta$. The $j$-th sample at node $i$ is defined as a tuple ${\mc{X}_{i,j},\mc{Y}_{i,j}} \subseteq \mathbb{R}^{784}
\times \{+1,-1\}$ and $\mb{b},c$ are the regression objective parameters (or the parameters of the separating hyperplane). This objective function is convex \cite{qureshi2021push} and for finite number of bounded-value data points clearly satisfies Assumption~\ref{ass_nonconv}. We run the algorithms over an exponential graph as a structured network to better show the convergence results. The link weights are designed WB to satisfy Assumption~\ref{ass_net}. Applying log-quantized data exchange with $\rho=\frac{1}{128}$ the dynamics also satisfies Assumption~\ref{ass_nonlin}. We compare our distributed  \textbf{HBNP-GT} algorithm with \textbf{GP} \cite{nedic2014distributed}, \textbf{SGP} \cite{spiridonoff2020robust,nedic2016stochastic}, \textbf{S-ADDOPT} \cite{qureshi2020s}, \textbf{ADDOPT} \cite{xi2017add}, and \textbf{PushSAGA} \cite{qureshi2021push} algorithms. For all algorithms, the step-sizes are hand-tuned for best performance and for \textbf{HBNP-GT} algorithm we have $\beta = 0.5$, $\alpha = 0.5$.
The optimality gaps (or the optimization residual) are compared
in Fig.~\ref{fig_mnist}, which clearly shows the accelerated convergence of our algorithm as compared with the existing results.    

\begin{figure}
	\centering
	\includegraphics[width=3.5in]{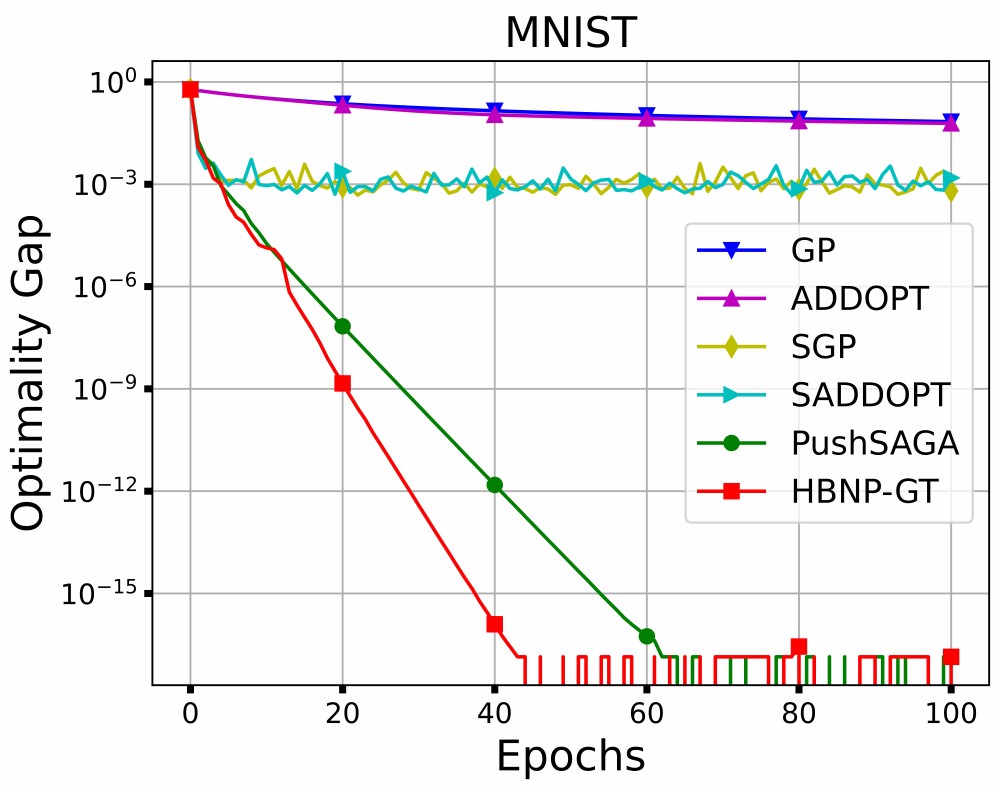}		 		
	\caption{This figure shows the logistic-regression convergence on the MNIST data-set for the proposed \textbf{HBNP-GT} algorithm as compared with some existing works .}   \label{fig_mnist}
\end{figure}

\section{Conclusions} \label{sec_con}
\subsection{Concluding Remarks}
This paper provides an accelerated momentum-based distributed algorithm for locally non-convex optimization over WB and possibly time-varying directed networks. The proposed method advances the state-of-the-art by addressing log-scale quantization, clipping, and general nonlinear sector-bound functions on data-sharing channels.
Our analytical and simulation results clearly show that the performance of the proposed algorithm does not degrade in the presence of these sector-bound nonlinearities. Different machine-learning applications are simulated to verify the results.

\subsection{Future Research}
The perturbation-based eigen-spectrum analysis in this work can be used to address (i) time-varying setups and (ii) general nonlinear sector-bound functions (e.g., log-scale quantization and clipping) for different distributed optimization, machine learning, and resource allocation applications. These nonlinear and time-varying models practically exist in most real-world distributed setups and may lead to sub-optimal convergence and even instability. As another direction of future research, one can extend the solution by applying Nesterov's momentum technique. Unlike the heavy-ball method in this paper, which uses a direct momentum term from the previous step, Nesterov's method incorporates a look-ahead gradient evaluation, which leads to improved convergence rate.

\section*{Acknowledgements} 
Authors would like to thank Usman A. Khan and Muhammad~I.~Qureshi for their insightful comments and for sharing their codes for MNIST data set classification.
	
This work is funded by Semnan University, research grant No. 226/1403/1403208.

\bibliographystyle{elsarticle-num-names}
\bibliography{bibliography}

\end{document}